\documentclass[preprint,12pt]{elsarticle}

\usepackage{graphicx}

\usepackage{amsmath}
\usepackage[notheorems]{definitions}
\usepackage{xspace}
\usepackage{comment}
\usepackage[dvips]{color}
\usepackage{pstricks}
\usepackage{pst-node}
\newcommand{\dkedit}[1]{{#1}}
\newcommand{\aiedit}[1]{{#1}}

\newcommand{\Omitthis}[1]{}
\newcommand{\suggestedcut}[1]{}
\newtheorem{theorem}{Theorem}
\newtheorem{lemma}{Lemma}
\newtheorem{definition}{Definition}
\newtheorem{example}{Example}

\newcommand{\Elts}{\ensuremath{E}\xspace}
\newcommand{\Feasible}{\ensuremath{{\cal F}}\xspace}
\newcommand{\SetSys}{\ensuremath{(\Elts,\Feasible)}\xspace}
\newcommand{\EltsP}{\ensuremath{E'}\xspace}
\newcommand{\FeasibleP}{\ensuremath{{\cal F'}}\xspace}
\newcommand{\OwnershipP}{\ensuremath{{\cal A'}}\xspace}
\newcommand{\SetSysP}{\ensuremath{(\EltsP,\FeasibleP)}\xspace}
\newcommand{\costs}{\ensuremath{\mathbf{c}}\xspace}

\newcommand{\mech}{\ensuremath{{\cal M}}\xspace}
\newcommand{\CHNASH}{\ensuremath{\nu}\xspace}
\newcommand{\ChNash}[1]{\ensuremath{\CHNASH(#1)}\xspace}
\newcommand{\MechPay}[2][]{\ensuremath{p_{#1}(#2)}\xspace}
\newcommand{\MechRat}[1]{\ensuremath{\phi_{#1}}\xspace}

\newcommand{\RP}{\ensuremath{\mathrm{AP}}\xspace}
\newcommand{\EP}{\ensuremath{\mathrm{MP}}\xspace}

\newcommand{\SSClass}{\ensuremath{{\mathcal{C}}}\xspace}
\newcommand{\Owned}[1]{\ensuremath{A^{#1}}\xspace}
\newcommand{\Owner}[1]{\ensuremath{o(#1)}\xspace}
\newcommand{\OwnerP}[1]{\ensuremath{\tilde{o}(#1)}\xspace}
\newcommand{\LEN}{\ensuremath{w}\xspace}
\newcommand{\len}[1]{\ensuremath{\LEN(#1)}\xspace}
\newcommand{\LPen}[2]{\ensuremath{D_{#1}(#2)}\xspace}
\newcommand{\Adj}[1]{\ensuremath{\beta(#1)}\xspace}
\newcommand{\Pay}[1]{\ensuremath{p_{#1}}\xspace}
\newcommand{\THBid}[1]{\ensuremath{B^{#1}}\xspace}
\newcommand{\RVCG}{\ensuremath{\mbox{RVCG}}\xspace}
\newcommand{\Ownership}{\ensuremath{{\cal A}}\xspace}
\newcommand{\SetSysA}{\ensuremath{(\SetSys,\Ownership)}\xspace}
\newcommand{\SetSysAP}{\ensuremath{(\SetSysP,\OwnershipP)}\xspace}

\journal{Artificial Intelligence Journal}

\begin{document}

\begin{frontmatter}



\title{False-name-proof Mechanisms for Hiring a Team\tnoteref{t1}}

\tnotetext[t1]{This paper is an extented version of 
  ``False-name-proof Mechanisms for Hiring a Team'' 
  in the Proceedings of the Workshop on Internet and Network Economics, 2007.}


\author[kyushu]{Atsushi Iwasaki\corref{cor1}}
\ead{iwasaki@inf.kyushu-u.ac.jp}
\author[usc]{David Kempe}
\ead{dkempe@usc.edu}
\author[usc]{Mahyar Salek}
\ead{salek@usc.edu}
\author[kyushu]{Makoto Yokoo}
\ead{yokoo@inf.kyushu-u.ac.jp}

\cortext[cor1]{Corresponding author}

\address[kyushu]{Graduate School of ISEE, Kyushu University, Fukuoka 819-0395, Japan}
\address[usc]{Department of Computer Science, University of Southern California, CA 90089-0781, USA}

\begin{abstract}
  We study the problem of hiring a team of selfish agents to perform a task.  
  Each agent is assumed to own one or more elements of a set system, 
  and the auctioneer is trying to purchase a feasible solution 
  by conducting an auction. 
  Our goal is to design auctions that are truthful and false-name-proof, 
  meaning that it is in the agents' best interest to reveal
  ownership of all elements (which may not be known to the auctioneer
  a priori) as well as their true incurred costs.
    
  We first propose and analyze a false-name-proof mechanism for the
  special case where each agent owns only one element in reality,
  \dkedit{but may pretend that this element is in fact a set of
    multiple elements}.  
  We prove that its frugality ratio is bounded by $2^n$, 
  which, \dkedit{up to constants}, matches a lower bound of $\Omega(2^n)$ 
  for all false-name-proof mechanisms in this scenario.
  We then propose a second mechanism \dkedit{for the general case in
    which agents may own multiple elements}.
  It requires the auctioneer to choose a reserve cost a priori, 
  and thus does not always purchase a solution. 
  In return, it is false-name-proof even
  when agents own multiple elements. 
  We experimentally evaluate the payment (as well as social surplus) of 
  the second mechanism through simulation.
\end{abstract}


\end{frontmatter}


\section{Introduction}
\label{sec:intro}
One of the important challenges of electronic commerce, 
in particular
in large-scale settings such as the Internet, is to design protocols 
for dealing with parties having diverse and selfish interests. 
Frequently, 
one of the most convenient ways of structuring these interactions 
is via \todef{auctions}: based on bids submitted by the participants, 
the auctioneer chooses whom to sell items to or purchase
\dkedit{items} from, and decides on appropriate payments. 
The analytical study of auctions for e-commerce has recently led to
very fruitful interactions between the fields of economics, game theory, 
theoretical computer science, and artificial intelligence. 

While single-item auctions have a long history of study in economics
(see, e.g., \cite{klemperer:guide,krishna:auction-theory}), 
the problem is significantly more complex 
when there are combinatorial dependencies between items. 
In a \todef{combinatorial auction}~\cite{combinatorial-auctions}, 
the auctioneer has a set of items for sale, 
and agents submit bids for different subsets. 
Each item can only be assigned to one agent.

\dkedit{%
In contrast to combinatorial auctions, where an auctioneer is trying
to \emph{sell} a set of items}, we study the problem of \todef{hiring a
team of agents} 
\cite{archer:tardos:path-mechanisms:J,garg:kumar:rudra:verma:J,talwar:price-of-truth}, 
In that problem, 
an auctioneer knows which subsets of agents can perform a complex task together, 
and needs to hire such a team. (called a \todef{feasible set} of agents.) 
Since the auctioneer does not know the true costs incurred by agents,
we assume that the auctioneer will use an auction to elicit bids.
A particularly well-studied special case of this problem is that of a \todef{path auction} 
\cite{archer:tardos:path-mechanisms:J,elkind:sahai:steiglitz,BeyondVCG,nisan:ronen:algorithmic:J}:
the agents own edges of a known graph, and the auctioneer wants to purchase an $s$-$t$ path. 

Selfish agents will try to maximize their profit, 
even if it requires misrepresenting their incurred cost or their identity. 
\dkedit{The field of \todef{mechanism design} focuses on the design of the
interaction between agents and computation to mitigate the effects
of such selfish behavior
\cite{nisan:ronen:algorithmic:J,mas-collel:whinston:green,papadimitriou:games}.}
In particular, there has been a lot of recent focus on 
the design of \todef{truthful} auctions, 
in which it is in the agents' best interest to reveal their true costs to the auctioneer. 


While the concept of truthfulness addresses the \dkedit{concern} that agents 
may misrepresent their true costs, there is a second way 
in which agents could cheat: 
an agent owning multiple elements of a set system (such as multiple edges in a graph) 
may choose different identities for interacting with the auctioneer, 
to obtain higher payments. 
Similarly, an agent owning one element may be able to pretend that 
this element is in fact a set of multiple elements, owned by different agents, 
to obtain payments for all of these ``pseudo-agents''. 
Such behavior is called \todef{false-name manipulation}, 
and was recently studied by Yokoo et al.~in the context of combinatorial auctions 
\cite{yokoo:aij2001,yokoo:geb:false-name}, where it was shown that for
any Pareto efficient auction, agents can profit by submitting bids as
two identities.

\subsection{Our contributions}\label{sec:contributions}
We introduce a model of false-name manipulation in auctions for hiring a team, 
such as $s$-$t$ path auctions.
In this model, the set system structure and element ownership 
are not completely known to the auctioneer.  
Thus, in order to increase profit, 
an agent who owns an element can pretend that the element is in fact 
a set consisting of multiple elements owned by different agents. 
Similarly, an agent owning multiple elements can submit bids 
for these elements under different identities. 
We call a mechanism \todef{false-name-proof} if it is truthful, 
and a dominant strategy is for each agent to reveal ownership of all elements.

Our first main contribution is a false-name-proof mechanism \EP 
for the special case in which each agent owns exactly one element.
\dkedit{Thus, the mechanism only needs to guard against an agent
pretending that a single element is a set of elements, owned by
distinct agents.}
This mechanism introduces an exponential multiplicative penalty 
against sets in the number of participating agents. 
We show that its frugality ratio 
(according to the definition of Karlin et al.~\cite{BeyondVCG}) is 
at most \dkedit{$2^n$} for all set systems of $n$ elements, 
which matches \dkedit{--- up to constants ---} a worst-case lower
bound of $\Omega(2^n)$ we establish for \emph{every} false-name-proof mechanism. 

When agents may own multiple elements, 
\dkedit{designing either a false-name
proof mechanism with bounded frugality ratio or proving an
impossibility result appears challening. The main reason is that we
currently do not have a good characterization of incentive-compatible
mechanisms with a sufficiently complex action space for the
agents. Instead,}
we present an alternative mechanism \RP, 
based on an a priori chosen reserve cost $r$ and additive penalties. 
The mechanism is false-name-proof in the general setting, 
but depends crucially on the choice of $r$, 
as it will not purchase a solution unless there is one 
whose cost (including the penalty) is at most $r$. 
We investigate the \RP mechanism experimentally 
for $s$-$t$ path auctions on random graphs, observing 
that \RP provides social surplus not too far from 
a Pareto efficient one at an appropriate reserve cost.
Also, the payments of \RP \dkedit{are} smaller than 
\dkedit{those of} the Vickrey-Clarke-Groves (VCG) mechanism when the
reserve cost is small, while \dkedit{they become} higher than VCG's
when the reserve cost is high.
However, the payment never exceeds the reserve cost.

\subsection{Related \dkedit{Work}}
\label{sec:related works}
Motivated by the need to deal with selfish users, 
there has been a large body of recent work 
at the intersection of game theory, economic theory 
and theoretical computer science~(see, e.g., \cite{nisan:selfish-agents,papadimitriou:games}). 
For instance, 
the seminal paper of Nisan and Ronen \cite{nisan:ronen:algorithmic:J}, 
which introduced mechanism design to the theoretical computer science community, 
studied the tradeoffs between agents' incentives and computational complexity. 
The loss of efficiency in network games 
due to selfish user behavior has been studied 
\dkedit{under the names} of ``price of anarchy'' 
(see, e.g., \cite{papadimitriou:games,roughgarden:tardos:selfish-routing:J}), 
and ``price of stability'' (see \cite{ADKTWR:J}).

The problem of hiring a team of agents in complex settings, 
at minimum total cost, has been shown to have many practical economic applications 
(see~\cite{feigenbaum:papadimitriou:sami:shenker:J,oneill:julian:chiang:boyd,lazar:orda:pendarakis:J,nisan:london:regev:carmiel} for examples).
In particular, the path auction problem has been 
the subject of a significant amount of prior research. 
The traditional economics approach to payment minimization 
(or profit maximization) is to construct the optimal Bayesian auction 
given the prior distributions from which agents' private values are drawn. 
Indeed, path auctions and similar problems have been studied recently 
from the Bayesian perspective in \cite{elkind:sahai:steiglitz,czumaj:ronen:expected}.
Here, we instead follow the approach pioneered 
by Archer, Tardos, Talwar and others~\cite{archer:tardos:path-mechanisms:J,bikchandani:devries:schummer:vohra,BeyondVCG,talwar:price-of-truth}, 
and study the problem from a worst-case perspective. 
Significant insight can be gained from an understanding of worst-case performance, 
and it enables an uninformed or only partially informed auctioneer 
to evaluate the trade-off between an auction tailored to assumptions 
about bidder valuations (which may or may not be correct) versus 
an auction designed to work 
as well as possible under unknown and worst-case market conditions.

If false-name bids are not a concern, 
then it has long been known that the VCG mechanism
\cite{vickrey:counterspeculation,clarke:multipart,groves:incentives}
gives a truthful mechanism and identifies the Pareto optimal solution. 
It is based on \aiedit{Vickrey's second-price auction}~\cite{vickrey:counterspeculation}, 
which is truthful for single-item auctions.
As the payments of VCG can be significantly higher than the cheapest
alternative solution, several papers
\cite{archer:tardos:path-mechanisms:J,talwar:price-of-truth,elkind:sahai:steiglitz,BeyondVCG}
have investigated the \todef{frugality} of mechanisms:
the overpayment compared to a natural lower bound. 
In particular, Karlin et al.~\cite{BeyondVCG} \dkedit{present} a mechanism 
--- called the $\sqrt{}$ mechanism --- achieving \dkedit{a} frugality ratio 
within a constant factor of optimal for $s$-$t$ path auctions in graphs.
Traditionally, for ``hiring a team'' auctions, 
\todef{incentive compatibility} has only \dkedit{encompassed making
the} revelation of true \emph{costs} a dominant strategy for each bidder.

The issue of false-name bids \dkedit{has been previously} studied in 
\dkedit{several cases of} combinatorial auctions \dkedit{and
  procurement auctions} by Yokoo et 
al.~\cite{iwasaki:dss:2005,suyama:aamas:2005,suyama:wine2005,yokoo:charactrization:ijcai2003,yokoo:aij2001},
who developed false-name-proof mechanisms in those scenarios, 
but also proved that no mechanism can be both false-name-proof and Pareto efficient.
Notice that the false-name-proof mechanisms for 
combinatorial procurement auctions given in \cite{suyama:aamas:2005,suyama:wine2005}
cannot be applied in our setting, 
as they assume additive valuations on the part of the auctioneer, 
i.e., that the auctioneer derives partial utility from partial solutions.
A somewhat similar scenario arises in job scheduling, 
where users may split or merge jobs to obtain earlier assignments. 
Moulin~\cite{moulin:split-proof} gives a mechanism that is truthful/strategy-proof
against both merges and splits and achieves efficiency within 
a constant factor of optimum. 
However, when agents can exchange money, 
no such mechanism is possible \cite{moulin:split-proof}.

For the specific case of path auctions, the impact of false-name bids
was studied by Du et al.~\cite{du:netecon:2006:J}. 
They showed that if agents can own multiple edges, 
then there is no false-name-proof and efficient mechanism. 
Furthermore, if bids are anonymous, 
i.e., agents do not report any identity for edge ownership, 
then no mechanism can be truthful/strategy-proof. 
Notice that this does not preclude false-name-proof and truthful mechanisms 
in which the auctioneer takes ownership of multiple edge 
by the same agent into account, and rewards the agent accordingly.

\suggestedcut{
The problems in the design of incentive-compatible mechanisms caused
by unknown domains have also been studied in the context of
combinatorial auctions by Babaioff et al.~\cite{babaioff:lavi:pavlov}. 
Among others, they show that useful monotonicity characterizations 
of incentive compatible mechanisms can break down 
in the case of unknown domains for individual agents.
}

\section{Preliminaries}
\label{sec:preliminaries}
%
We begin by defining formally the framework for auctions to hire a team.
Our framework is based on that of 
\cite{archer:tardos:path-mechanisms:J,bikchandani:devries:schummer:vohra,BeyondVCG,talwar:price-of-truth}.
A \todef{set system} \SetSys is specified by 
a set \Elts of $n$ \todef{elements} and 
a collection $\Feasible \subseteq 2^{\Elts}$ of \todef{feasible sets}.
For instance, in the important special case of an $s$-$t$ path
auction, $S \in \Feasible$ if and only if $S$ is an $s$-$t$ path.
\dkedit{We are only interested in set systems that are \emph{monopoly-free},
in the sense that $\bigcap_{S \in \Feasible} S = \emptyset$, i.e., no
agent is in all feasible sets.}

In previous work on ``hiring a team'' auctions, 
each element $e$ was associated with a different selfish agent. 
Here, we depart from this assumption, 
in that an agent may own \emph{multiple} elements. 
$\Owned{i}$ denotes the set of elements owned by agent $i$,
\dkedit{which is} an element of a partition $\Ownership$ of \Elts.
An \todef{owned set system}, i.e., a set system with ownership structure, 
is specified by $(\SetSys,\Ownership)$.
\dkedit{We use $\Owner{e}$ to denote the owner of element $e$, i.e.,
the unique $i$ such that $e \in \Owned{i}$.}
Each element $e$ has an associated \todef{cost} $c_e$, the true cost
that its owner $\Owner{e}$ will incur if $e$ is selected by the
mechanism.\footnote{For costs, bids, etc., 
we extend the notation by writing $c(S) = \sum_{e \in S} c_e$, $b(S) = \sum_{e \in S} b_e$, etc.}
This cost is \todef{private}, i.e., known only to \Owner{e}.
An \todef{auction} consists of two steps:

\begin{enumerate}
\item Each agent $i$ submits sealed bids $(b_e, \OwnerP{e})$ for
  elements $e$, where \OwnerP{e} denotes the identifier of $e$'s
  purported owner \aiedit{which need not be the actual owner. 
  (However, no agent $i$ can claim ownership of an element $e$
  owned by another agent $i' \neq i$.)}
\item Based on the bids, the auctioneer uses an algorithm that is
  common knowledge among the agents in order to select a feasible set 
  $S^* \in \Feasible$  as the winner and compute a payment \Pay{i} for
  each agent $i$ with an element $e$ such that $i = \OwnerP{e}$. 
  We say that the elements $e \in S^*$ \todef{win}, and all other elements
  \todef{lose}. 
\end{enumerate}
The \todef{profit} of an agent $i$ is the sum of all payments she receives, 
minus the incurred cost $c(S^* \cap \Owned{i})$. Each agent is only
interested in maximizing her profit, and might choose to misrepresent
ownership or costs to this end. However, we assume that agents do not
collude.
%
%
Past work on incentive compatible mechanisms has focused on
\todef{truthful} mechanisms. That is, the assumption was that each
agent $i$ submits bids only for elements $e \in \Owned{i}$ she actually
owns, and reports correct ownership $\Owner{e} = i$ for all of them. 
If agents report correct ownership for all $e \in \Owned{i}$, then
a mechanism is truthful by definition if for any fixed
vector $b^{-i}$ of bids by all agents other than $i$, it is in agents
$i$'s best interest to bid $b_e = c_e$ for all $e \in \Owned{i}$, 
i.e., agent $e$'s profit is maximized by bidding $b_e = c_e$ for all
these elements $e$. 

In this paper, we extend the study of truthful mechanisms 
to take into account \todef{false-name manipulation}: 
agents claiming ownership of non-existent elements 
(which we call \todef{self-division}) or 
choosing not to disclose ownership of elements 
(which we call \todef{identifier splitting}).
Identifier Splitting is the most natural form of false-name bidding on
the part of an agent, and the one studied in the past for
combinatorial auctions, by Yokoo et al.~\cite{yokoo:aij2001,yokoo:geb:false-name}. 
The notion of self-division is motivated by graph-theoretic problems
(such as shortest paths), when there is uncertainty on the part of the
auctioneer about the underlying set system.

%
%
\begin{definition}[Identifier Splitting \cite{yokoo:aij2001,yokoo:geb:false-name}]
An agent $i$ owning a set \Owned{i} may choose to use 
different identifiers in her bid for some or all of the elements. 
Formally, the owned set system \SetSysA is replaced 
by \ensuremath{(\SetSys,\OwnershipP)}\xspace, 
where $\Owned{'} = A \setminus \{\Owned{i}\} \cup \{\Owned{i'}\} \cup \{\Owned{i''}\}$, 
and $\Owned{i} = \Owned{i'} \cup \Owned{i''}$ 
when agent $i$ uses two identifiers $i'$ and $i''$. 
\end{definition}

\begin{definition}[Self-Division]
An agent $i$ owning element $e$ is said to \todef{self-divide} $e$ 
if $e$ is replaced by two or more elements $e_1, \ldots, e_k$, 
and different owners are reported for the $e_i$. Formally, 
the owned set system \SetSysA is replaced by \SetSysAP, 
whose elements are $\EltsP = \Elts \setminus \SET{e} \cup \SET{e_1, \ldots, e_k}$, 
such that the feasible sets \FeasibleP are exactly those sets $S$
not containing $e$, as well as sets 
$S \setminus \SET{e} \cup \SET{e_1, \ldots, e_k}$ 
for all feasible sets $S \in \Feasible$ containing $e$. 
The ownership structure is $\Owned{i_j} = \SET{e_j}$ for $j=1, \ldots, k$, 
where each $i_j$ is a new agent.
\end{definition}

Intuitively, self-division allows an agent to pretend that multiple
distinct agents are involved in doing the work of element $e$, and
that each of them must be paid separately.
For self-division to be a threat, 
there must be uncertainty on the part of the auctioneer 
about the true set system \SetSys. 
In particular, it is meaningless to talk about a mechanism for an
individual set system, as the auctioneer does not know a priori what
the set system is. Hence, we define 
\todef{classes of set systems closed under subdivision}, 
as the candidate classes on which mechanisms must operate.
\begin{definition}

\begin{enumerate}
\item For two set systems \SetSys and \SetSysP, we say \SetSysP is 
\todef{reachable} from \SetSys by subdivisions 
if \SetSysP is obtained by (repeatedly) replacing
individual elements $e \in E$ with $\SET{e_1, \ldots, e_k}$, such that
the feasible sets \FeasibleP are exactly those sets $S$
not containing $e$, as well as sets 
$S \setminus \SET{e} \cup \SET{e_1, \ldots, e_k}$ for all
feasible sets $S \in \Feasible$ containing $e$. 

\item A class \SSClass of set systems is \todef{closed under
    subdivisions} iff with \SetSys, all set systems reachable from
  \SetSys by subdivisions are also in \SSClass.
\end{enumerate}

\end{definition}
For example, $s$-$t$ path auction set systems are closed under subdivisions,
whereas minimum spanning tree set systems are not (because 
subdivisions would introduce new nodes that must be spanned).
On the other hand, minimum Steiner tree set systems 
with a fixed set of terminals are susceptible to 
false-name manipulation.

In both identifier splitting and self-division, 
we will sometimes refer to the new agents $i'$ 
whose existence $i$ invents as \todef{pseudo-agents}. 
A mechanism is \todef{false-name-proof} if it is a dominant strategy 
for each agent $i$ to simply report the pair $(c_e, i)$ as a bid for 
each element $e \in \Owned{i}$. 
Thus, neither identifier splitting nor self-division 
nor bids $b_e \neq c_e$ can increase the agent's profit. 
Among other things, this allows us to use $b_e$ and $c_e$ interchangeably 
when discussing false-name-proof mechanisms. 
Notice that we explicitly define the concept of false-name-proof mechanisms 
to imply that the mechanism is also truthful when each agent $i$ owns only one element.
 
\subsection{Efficiency and Frugality}
In designing and analyzing a mechanism for hiring \dkedit{a team}, 
there are several other desirable properties 
besides being false-name-proof (or at least truthful). 
Two particularly important ones are efficiency and frugality.
%
%
A mechanism is \todef{Pareto efficient} if it always maximizes 
the sum of all participants' utilities (including that of the auctioneer). 
This maximizes social surplus. 
In the case of hiring a team, the auctioneer's utility is exactly $-\sum_i \Pay{i}$, 
the negative of the sum of all payments. 
Hence, all payments cancel out, and a mechanism is Pareto efficient 
if and only if it always purchases the cheapest team or $s$-$t$ path.
While it is well-known that the VCG mechanism is 
truthful and Pareto efficient~\cite{vickrey:counterspeculation,clarke:multipart,groves:incentives},
Du et al.~\cite{du:netecon:2006:J}
show that there is no Pareto efficient and 
false-name-proof mechanism, even for $s$-$t$ path auctions.
Yokoo et al.~\cite{yokoo:geb:false-name} showed the same for
combinatorial auctions. 

While Pareto efficient mechanisms maximize social welfare, 
they can significantly overpay compared to other mechanisms~\cite{BeyondVCG,elkind:sahai:steiglitz}.
In order to analyze the overpayment, 
we use the definition of \todef{frugality ratio}
from~\cite{BeyondVCG}.
The idea of the frugality ratio is to compare
the payments to a ``natural'' lower bound, generalizing the idea of
the second lowest cost. (It is easy to observe that no meaningful ratio
is possible when comparing to the actual lowest cost.)

\begin{definition}[\cite{BeyondVCG}] \label{def:chnash}
Let \SetSys be a set system, and
\dkedit{\costs a cost vector for the elements}.
Let $S$ be a cheapest feasible set with respect to the $c_e$ 
(where ties are broken lexicographically).
We define \ChNash{\costs} to be the solution to the following optimization problem.
\[ 
\begin{array}{l}
\mbox{Minimize } \sum_{e\in S} x_e
\mbox{ subject to} \\[1ex]
 (1) \; x_e \geq c_e \;\; \mbox{ for all } $e$\\[1ex]
 (2) \; x(S \setminus T) \leq c(T \setminus S)
         \;\; \mbox{ for all } T \in \Feasible\\[1ex]
 (3) \; \mbox{For every $e\in S$, there is a $T_e \in \Feasible$ such that}\\
         \phantom{(3) \; }e \notin T_e \mbox{ and }
         x(S \setminus T_e) = c(T_e\setminus S)
\end{array} 
\]
\end{definition}

This definition essentially captures the payments 
in a ``cheapest Nash equilibrium'' of a first-price auction, 
and gives a natural lower bound generalizing second-lowest cost for comparison purposes. 

\begin{definition} \label{def:frugality}
The frugality of a mechanism \mech for a set system \SetSys is 
\begin{eqnarray*}
\MechRat{\mech}
& = & \sup_{\costs} \frac{\MechPay[\mech]{\costs}}{\ChNash{\costs}},
\end{eqnarray*}
i.e., the worst case, over all cost vectors \costs, of the overpayment
compared to the ``first-price'' payments.
Here, \MechPay[\mech]{\costs} denotes the total payments made by \mech
when the cost vector is \costs.
\end{definition}

\section{A Multiplicative Penalty Mechanism}
\label{sec:division}
In this section, 
we focus on a mechanism \EP with \emph{multiplicative penalties}, 
\dkedit{as well as lower bounds,}
for arbitrary ``hiring a team'' instances. 
The \EP mechanism always buys a solution, 
and \dkedit{so long as each agent owns} one element only,
it is false-name proof.\footnote{In fact, \EP works 
even if an agent owns multiple elements, 
so long as all of these elements are required at the same time. 
In other words, if we can consider a set of elements 
as a virtual single element, \EP is false-name-proof.}
We analyze the frugality ratio of \EP for arbitrary instances, 
and prove that it is \dkedit{at most $2^n$}, 
matching \dkedit{--- up to constants ---} a lower bound of $\Omega(2^n)$ 
for any false-name-proof mechanism. 
\subsection{The Mechanism \EP} \label{sec:GUSC}
\dkedit{The} mechanism \EP is based on exponential multiplicative penalties. 
It is false-name-proof for arbitrary classes of set systems 
closed under subdivisions, 
\emph{so long as each agent only owns one element} 
\dkedit{(In other words, it guards against self-division by agents)}. 
We can therefore identify elements $e$ with agents.
Since we assume \dkedit{that} each agent owns exactly one element, 
\Ownership is automatically determined by \Elts, so we can focus on
set systems instead of owned set systems.

After the agents submit bids $b_e$ for elements, 
\EP chooses the set $S^*$ minimizing $b(S) \cdot 2^{\SetCard{S}-1}$, among
all feasible sets $S \in \Feasible$. Each agent $e \in S^*$ is then
paid her threshold bid 
$2^{\SetCard{S^{-e}}-\SetCard{S^*}} b(S^{-e}) - b(S^* \setminus \SET{e})$,
where $S^{-e}$ \dkedit{denotes} the best solution (with respect to the objective
function $b(S) \cdot 2^{\SetCard{S}-1}$) among feasible sets $S$ not
containing $e$. 
Notice that while this selection may be NP-hard in general, it
can be accomplished in polynomial time for path auctions, by using
the Bellman/Ford algorithm to compute the shortest path for each
number of hops, and then \dkedit{choosing from} the at most $n$ such
shortest paths.

\begin{theorem}
\label{thm:upper bound}
For all classes of set systems closed under subdivision, 
\EP is false-name-proof, so long as each agent only owns one
element. Furthermore, it has frugality ratio 
\dkedit{$O(2^n)$}, where $n=\SetCard{\Elts}$.
\end{theorem}

\begin{proof}
If an agent $e=e_0$ self-divides into $k+1$ elements $e_0, \ldots,
e_k$, then either all of the $e_i$ or none of them are included in any
feasible set $S$. Thus, we can always think of just one threshold
$\tau_k(e)$ for the self-divided agent $e$: if the sum of the bids of
all the new elements $e_j$ exceeds $\tau_k(e)$, then $e$ loses;
otherwise, it is paid at most $(k+1) \tau_k(e)$. 
The original threshold of agent $e$ is $\tau(e) = \tau_0(e)$.

The definition of the \EP mechanism implies that
$\tau_k(e) \leq 2^{-k} \tau(e)$. 
If $e$ still wins after self-division
(otherwise, there clearly is no incentive to self-divide), the total
payment to $e$ is at most $(k+1) 2^{-k} \tau(e)$. 
The alternative of not self-dividing, and submitting a bid of 0, 
yields a payment of $\tau(e) \geq (k+1) 2^{-k} \tau(e)$. 
Thus, refraining from self-division is a dominant strategy. 
Given that no agent will submit false-name bids, 
the monotonicity of the selection rule implies that the
mechanism is incentive compatible, 
and we can assume that $b_e = c_e$ for all agents $e$. 

To prove the upper bound on the frugality ratio, 
consider again any winning agent $e \in S^*$. Her threshold bid is
\begin{eqnarray*}
\tau(e) 
& = & \min_{T \in \Feasible: e \notin T} 2^{\SetCard{T}-\SetCard{S^*}} c(T)
      - c(S^* \setminus \SET{e}),
\end{eqnarray*}
and the total payment is the sum of individual thresholds for $S^*$,
\begin{eqnarray*}
\MechPay[\EP]{\costs}
& = & \sum_{e \in S^*} \min_{T \in \Feasible: e \notin T} 
      2^{\SetCard{T}-\SetCard{S^*}} c(T)
      - c(S^* \setminus \SET{e}) \\
& \leq & 2^{\dkedit{n-\SetCard{S^*}}} \sum_{e \in S^*} \min_{T \in \Feasible: e \notin T} c(T).
\end{eqnarray*}

To obtain \dkedit{the frugality ratio} from this upper bound on the payments,
we need a lower bound on the value \ChNash{\costs} 
(see Definition~\ref{def:frugality}). 
Let $S$ be the cheapest solution with respect to the $c_e$, i.e.,
without regard to the sizes of the sets.
By Definition \ref{def:chnash}, 
$\ChNash{\costs} = \sum_{e \in S} x_e$, subject to the constraints of
the mathematical program given.
Focusing on any fixed agent $e'$, we let $T_{e'}$ denote the set from
the third constraint of Definition \ref{def:chnash}, and can rewrite
\begin{equation}
  \begin{array}{lclclcl}
    \ChNash{\costs}
    & = & \sum_{e \in S \setminus T_{e'}} x_e + \sum_{e \in S \cap T_{e'}} x_e\\
    & = & \sum_{e \in T_{e'} \setminus S} c_e + \sum_{e \in T_{e'} \cap S} x_e
    & \geq & c(T_{e'}).
  \end{array} \label{eqn:nash-bound}
\end{equation}

Since this inequality holds for all $e'$, 
we have proved that $\ChNash{\costs} \geq \max_{e \in S} c(T_e)$. 
On the other hand, we can further bound the payments by
\begin{eqnarray*}
2^{\dkedit{n-\SetCard{S^*}}} 
\sum_{e \in S^*} \min_{T \in \Feasible: e \notin T} c(T)
& \leq & \dkedit{\SetCard{S^*}} \cdot 2^{\dkedit{n-\SetCard{S^*}}} 
         \cdot \max_{e \in S^*} \min_{T \in \Feasible: e \notin T} c(T)\\
& \leq & \dkedit{\frac{\SetCard{S^*}}{2^{\SetCard{S^*}}}} \cdot 2^n 
         \cdot \max_{e \in S} \min_{T \in \Feasible: e \notin T} c(T)\\
& \leq & 2^n \cdot \max_{e \in S} c(T_e).
\end{eqnarray*}

Here, the \dkedit{middle} inequality followed because for all 
$e \in S^* \setminus S$, the minimizing set $T$ is actually equal to
$S$, and therefore cannot have larger cost than $c(T_e)$ for any 
$e \in S$, by definition of $S$.
Thus, the frugality ratio of \EP is
\[ 
\MechRat{\EP} \; = \; \sup_{\costs} \frac{\MechPay[\EP]{\costs}}{\ChNash{\costs}}
\; \leq \; \frac{2^n \max_{e \in S} c(T_e)}{\max_{e \in S} c(T_e)} 
\; = \; 2^n. 
\]
\end{proof}
\subsection{An Exponential Lower Bound}
An exponentially large frugality ratio is not desirable. 
Unfortunately, any mechanism which is false-name-proof will have to
incur such a penalty, as shown by the following theorem.

\begin{theorem} \label{thm:lower-bound}
Let \SSClass be any class of monopoly free set systems 
closed under subdivisions, 
and \mech be any truthful and false-name-proof mechanism for \SSClass. 
Then, the frugality ratio of \mech on \SSClass is $\Omega(2^n)$ 
for set systems with $\SetCard{E} = n$.
\end{theorem}

\begin{proof}
Let $(E_0, \Feasible_0) \in \mathcal{C}$ be a set system minimizing
$\SetCard{E_0}$. 
Let $S^* \in \Feasible_0$ be the winning set under \mech winning 
when all agents $e \in E_0$ bid 0, 
and let $e \in S^*$ be arbitrary, but fixed.
Because $(E_0, \Feasible_0)$ is monopoly free, there must be 
a feasible set $T \in \Feasible_0$ with $e \notin T$ and $T \not\subseteq S^*$. 
Among all such sets $T$, let $T_e$ be one minimizing $\SetCard{S^* \cup T}$, 
and let $\hat{e}$ in $T_e \dkedit{\setminus S^*}$ be arbitrary. 
Define $Z = (T_e \cup S^*) \setminus \SET{e,\hat{e}}$ (the ``zero bidders''), 
and $I = E_0 \setminus (T_e \cup S^*)$ (the ``infinity bidders'').
Consider the following bid vector: both $e$ and $\hat{e}$ bid $1$, 
all agents $e' \in Z$ bid $0$, and all agents $e' \in I$ bid $\infty$. 
Let $W$ be the winning set. 
We claim that $W$ must contain at least one of $e$ and $\hat{e}$ 
(w.l.o.g., assume that $e \in W$). 
For $W$ cannot contain any of the infinity bidders. 
%
And if it contained neither $e$ nor $\hat{e}$, 
then $W$ would have been a candidate for $T_e$ 
with smaller $\SetCard{W \cup S^*}$, which would contradict the choice of $T_e$. 

Now, let $(E_k, \Feasible_k)$ be the set system resulting if 
agent $e$ self-divides into new agents $e_0, \ldots, e_k$, for $k \geq 0$. 
\suggestedcut{
Thus, all agents $e' \neq e$ stay the same, 
and the feasible sets are exactly those $S \in \Feasible_0$ with $e \notin S$, 
and those $S$ with $\SET{e_0, \ldots, e_k} \subseteq S$ 
and $S \setminus \SET{e_0, \ldots, e_k} \cup \SET{e} \in F_0$. 
Let $(E_k, \Feasible_k)$ denote the set system resulting 
from $k$-wise self-division (for $k \geq 0$). 
}
Define $\tau(j,k)$, for $j=0, \ldots, k$, to be the threshold bid 
under \mech for agent $e_j$ in the set system $(E_k,\Feasible_k)$, 
given that all $e' \in Z$ bid 0, all $e' \in I$ bid $\infty$, 
and all $e_i$ for $i \neq j$ also bid 0, while $\hat{e}$ bids 1. 
Above, we thus showed that $1 \leq \tau(0,0) < \infty$. 
We now show by induction on $d$ that for all $d$, 
there exists an $h \leq d$ such that
\begin{eqnarray*}
2^{-d} \sum_{i=0}^{k} \tau(i,k) 
& \geq & \sum_{i = h}^{k+h}\tau(i, k+d).
\end{eqnarray*}

The base case $d=0$ is trivial. For the inductive step, 
assume that we have proved the statement for $d$. 
Because $\mech$ is truthful,
the payment of an agent is exactly equal to the threshold bid, so each
agent $i$ is paid $\tau(i,k+d)$ in the auction on the set system
$(E_{k+d},\Feasible_{k+d})$ with the bids as given above.
If agent $i$ were to self-divide into two new agents, 
the new set system would be $(E_{k+d+1},\Feasible_{k+d+1})$, 
and the payment of agent $i$ 
(who is now getting paid as two pseudo-agents $i$ and $i+1$) 
would be $\tau(i,k+d+1) + \tau(i+1,k+d+1)$. 
Because $\mech$ was assumed to be false-name-proof, 
it is not in the agent's best interest to self-divide in such a way, 
i.e., $\tau(i, k+d) \geq \tau(i, k+d+1) + \tau(i+1, k+d+1)$. 
Summing this inequality over all agents $i=h,\ldots,h+k$, we obtain
\begin{eqnarray*}
\sum_{i=h}^{h+k} \tau(i, k+d) 
& \geq & 
\sum_{i=h}^{h+k} (\tau(i, k+d+1) + \tau(i+1, k+d+1))\\
& = & 
\sum_{i=h}^{h+k} \tau(i, k+d+1) + \sum_{i=h+1}^{h+k+1} \tau(i, k+d+1). 
\end{eqnarray*}

Define $\ell = 0$ if 
$\sum_{i=h}^{h+k} \tau(i, k+d+1) \leq \sum_{i=h+1}^{h+k+1}\tau(i, k+d+1)$;
otherwise, let $\ell = 1$.
Then, the above inequality implies that
\begin{eqnarray*}
\sum_{i=h}^{h+k} \tau(i, k+d) 
& \geq & 2\sum_{i=h+\ell}^{h+k+\ell} \tau(i, k+d+1).
\end{eqnarray*}
Finally, setting $h' := h + \ell$, we can combine this inequality with
the induction hypothesis to obtain that
\begin{eqnarray*}
2^{-(d+1)} \sum_{i=0}^{k} \tau(i,k) 
& \geq & \sum_{i = h'}^{k+h'}\tau(i,k+d+1),
\end{eqnarray*}
which completes the inductive proof.

Applying this equation with $k=0$, 
we obtain that for each $d \geq 0$, 
there exists an $h \leq d$ such that $\tau(h,d) \leq 2^{-d} \cdot \tau(0,0)$. 
Thus, in the set system $(E_d, \Feasible_d)$, 
if all infinity bidders have cost $\infty$, 
agent \dkedit{$e_h$} has cost just above $2^{-d} \tau(0,0)$, 
and all other agents have cost $0$, 
then agent $\hat{e}$ must be in the winning set, 
and must be paid at least 1. 
But it is easy to see that in this case, $\nu(c) = 2^{-d} \tau(0,0)$, 
and the frugality ratio is thus at least $2^d/\tau(0,0) = \Omega(2^d)$
(since $\tau(0,0)$ is a constant independent of $d$). 
Finally, $\SetCard{E_d} = \SetCard{Z} + \SetCard{I} + d \dkedit{+ 2}$, and
because $Z$ and $I$ are constant for our class of examples, 
the frugality ratio is 
$2^{-(\SetCard{Z}+\SetCard{I}\dkedit{+2})} \cdot 2^n /\tau(0,0) = \Omega(2^n)$.
\end{proof}

In this section, we presented the \EP mechanism based on multiplicative penalties. 
\EP always buys a feasible set. 
However, \EP is \dkedit{guaranteed to be} false-name-proof
\dkedit{only} when each agent \dkedit{owns a single} element.
At this point, we do not know if there exist any false-name-proof mechanisms 
against identifier splitting which \emph{always} buy a set at finite cost. 
This is an intriguing open question for future work. 

\section{An Additive Penalty Mechanism with Reserve Cost}\label{sec:splitting}
We next propose another false-name-proof mechanism \RP based on
additive penalties and a reserve cost. 
The mechanism requires no assumption on whether agents have single or
multiple elements in a set system, and we will prove that it is
false-name-proof 
even when agents own multiple elements. 
However, \RP does not always purchase a feasible set; it requires
the auctioneer to decide on a \todef{reserve cost}, and will only
purchase a solution if there is a feasible solution whose cost
(including penalties) does not exceed the reserve cost. 

We can interpret the reserve cost as an upper bound on the cost
(including penalties) the auctioneer is willing to pay. 
This is particularly reasonable if we assume that the auctioneer
already has a way of performing the task using a single agent of cost $r$, 
such as a direct \dkedit{edge $(s,t)$} with cost $r$ in a network.
If the bids by agents are such that the auctioneer chooses this alternative, 
then none of the agents (including the auctioneer) receives positive utility.
Clearly, the right choice of the reserve cost $r$ will be crucial for
the performance of the mechanism.
\subsection{The \RP mechanism}
\suggestedcut{
We next propose and analyze a mechanism called \RP, based on additive
penalties and a reserve cost. It will only purchase a solution when
the total cost (including penalties) does not exceed the a priori
chosen reserve cost $r$, and thus requires a judicious choice of $r$
by the auctioneer. 
In return, \RP is false-name-proof even when agents own multiple elements. 
}

The \RP mechanism is based on adding to the reported costs of the agents 
a penalty growing in the number of agents participating in a solution. 
For any set $S\in \Feasible$, let \len{S} denote the number of 
(pseudo-)agents owning one or more elements of $S$, called 
the \todef{width} of the set $S$.
The width-based \todef{penalty} for a set $S$ of width \len{S} 
is $\LPen{r}{\len{S}} = \dkedit{(1-2^{1-\len{S}})} \cdot r$.
Based on the \dkedit{reported} costs and the penalty, 
we define the \todef{adjusted cost} of a set $S$ to be
$\Adj{S} = b(S) + \LPen{r}{\len{S}}$. 

The \RP mechanism first determines the set $S^*$ minimizing 
the adjusted cost \Adj{S}, 
among all feasible sets $S\in{\Feasible}$. 
If its adjusted cost exceeds the reserve cost $r$, 
then \RP does not purchase any set, and does not pay any agents. 
Otherwise, it chooses $S^*$, and pays each winning agent 
(i.e., each agent $i$ with $S^* \cap \Owned{i} \neq \emptyset$) 
her threshold bid 
\begin{eqnarray*}
\Pay{i} & = & \min (r, \Adj{S^{-i}}) 
- \big(b(S^{*}\setminus \Owned{i}) + \LPen{r}{\len{S^*}}\big)
\end{eqnarray*}
with respect to \Adj{S}. Here, 
$S^{-i}$ denotes the best solution with respect to \Adj{S} such
that $S^{-i}$ contains no elements from \Owned{i}. 
\suggestedcut{
Notice that the set
$S^*$ selected by \RP will frequently differ from that selected by
VCG, as \RP takes the number of agents in a solution into account.
}

Notice that if we assume \dkedit{that}
the auctioneer requires an additional cost 
of $(1-2^{1-\len{S}}) \cdot r$ for handling a team $S$,
then \RP is identical to the VCG mechanism with reserve cost $r$, 
since the adjusted cost becomes the true total cost 
(including the additional cost of the auctioneer).
Thus, if we assume that there exists no false-name manipulation,
it is natural that \RP is incentive compatible
since it is one instance of VCG. 

\begin{example}{}
Consider the example in Figure~\ref{fig:AP-example}. Assume that
the reserve cost is $r=10$. If agent $X$ does not split identifiers,
the adjusted cost of the path $s$-$v$-$t$ is 2 (since it only involves
one agent, the penalty is 0), and the adjusted cost of the edge
$s$-$t$ is 8. Thus, the payment to agent $X$ is 8.

\begin{figure}[htb]
  \begin{center}
    \psset{unit=1.0cm} \pspicture(-2.5,-0.7)(2.5,1.3)
    \cnodeput(-2,0){s}{$s$} \cnodeput(2,0){t}{$t$}
    \cnodeput(0,0){v}{$v$} 

    \ncline{->}{s}{v}\Bput{1($X$)} 
    \ncline{->}{v}{t}\Bput{1($X$)}
    \ncarc[arcangleA=40,arcangleB=40]{->}{s}{t}\Aput{8($Y$)}
    \endpspicture
    \caption{An example of \RP.} \label{fig:AP-example}
  \end{center}
\end{figure}

If agent $X$ instead uses two different identifiers $X'$ and $X''$ for
the two edges, the penalty for the path $s$-$v$-$t$ is $10/2 = 5$.
Thus, while the path still wins, the payment to each of $X'$ and $X''$
is now $8- (1+5) = 2$, so the total payment to agent $X$ via
pseudo-agents is 4. In particular, agent $X$ has no incentive to split
identifiers in this case.
\end{example}

\subsection{Analysis of \RP}
In this section, 
we prove that simply submitting the pair $(b_e, i)$ 
for each element $e \in \Owned{i}$ is a dominant strategy for each agent $i$ 
under the mechanism \RP. 
Furthermore, we prove that the payments of the \RP mechanism never
exceed $r$. As a first step, we prove that it never increases an
agent's profit to engage  in identifier splitting.

\begin{lemma}\label{lem:identifier-splitting}
Suppose that agent $i$ owns elements \Owned{i}, and splits identifiers into
$i', i''$, with sets $\Owned{i'}, \Owned{i''}$, such
that $\Owned{i'} \cup \Owned{i''} = \Owned{i}$. Then, the profit agent
$i$ obtains after splitting is no larger than that obtained before
splitting.
\end{lemma}

\begin{proof}
Let $S^*\in \Feasible$ be 
the winning set prior to agent $i$'s identifier split.
We first consider the case when the winning set does not change 
due to the identifier split.
If only one of the new pseudo-agents $i', i''$ wins (say, $i'$), then
$\Adj{S^{-i'}} \leq \Adj{S^{-i}}$, because every
feasible set not using elements from \Owned{i} also does not use
elements from \Owned{i'}. Hence, the payment of $i$ could only
decrease, and we may henceforth assume that both $i'$ and $i''$ win, 
which means that the width of the winning set $S^*$ increases from
\dkedit{\len{S^*} to $\len{S^*}+1$}.

For simplicity, we write $\THBid{-i} = \min(r, \Adj{S^{-i}})$, 
and similarly for $i'$ and $i''$.
The payment to $i$ before the split is
$\THBid{-i} - (b(S^* \setminus \Owned{i}) + \LPen{r}{w})$,
whereas the new payment after the split is
\begin{eqnarray*}
\lefteqn{\THBid{-i'} - (b(S^* \setminus \Owned{i'}) + \LPen{r}{w+1})
      + \THBid{-i''} - (b(S^* \setminus \Owned{i''}) + \LPen{r}{w+1})}\\
    & = & \THBid{-i'} + \THBid{-i''} - 2b(S^*) + b(S^* \cap \Owned{i}) 
    -2\LPen{r}{w+1}.
\end{eqnarray*}

As argued above, we have that $\THBid{-i''} \leq \THBid{-i}$, 
and by definition of $\THBid{-i'}$, we also know that $\THBid{-i'} \leq r$. 
Thus, canceling out penalty terms, the increase in payment to agent
$i$ is bounded from above by 
\begin{eqnarray*}
  \THBid{-i'} + \THBid{-i''} - \THBid{-i} - b(S^*) -r
  & \leq & r + \THBid{-i} - \THBid{-i} - b(S^*) - r \\
  & = & -b(S^*) \\
  & \leq & 0.
\end{eqnarray*}

Hence, identifier splitting can only lower the payment of agent $i$.
Since the total cost incurred by agent $i$ stays the same, this proves
that there is no benefit in identifier splitting.

Next, suppose that the winning set after the split changes to 
$S'^{*} \neq S^*$.
Clearly, if $i$ does not win at all after the split, i.e., 
$S'^* \cap \Owned{i} = \emptyset$, then $i$ has no incentive to split
identifiers. Otherwise, if $i$ does win after the split, then $i$ must
also win before the split. For the split can only increase
\LPen{r}{\len{S}} for all sets $S$ containing any of $i$'s elements,
while not affecting \LPen{r}{\len{S}} for other sets. 
We can assume w.l.o.g.~that agent $i$ bids $\infty$ on all elements 
$e \in \Owned{i} \setminus S'^*$. For the winning set will stay the
same, because \Adj{S'^*} stays the same, and \Adj{S} can only increase
for other sets $S$, and the payments can only increase.

But then, $S'^*$ will also be the winning set if $i$
does not split identifiers 
(the adjusted cost \Adj{S'^*} decreases, while all other
adjusted costs stay the same). Now, we can apply the argument from
above to show that the payments to agent $i$ do not increase as a
result of splitting identifiers.
Thus, so long as an agent can submit bids of false cost instead, it is
never a dominant strategy to split identifiers. 
\end{proof}

Lemma \ref{lem:identifier-splitting} can be extended naturally to deal
with $k$-way identifier splitting.
Notice that the proof also shows that \RP is false-name-proof
against self-division. 

\begin{theorem} \label{thm:false-name-proof}
For all classes of set systems closed under subdivision, 
\RP is false-name-proof, even if agents can own multiple elements and
split identifiers. 
Thus, for each agent $i$, submitting bids $(c_e, i)$ 
for each element $e \in \Owned{i}$ is a dominant strategy.
\end{theorem}

\begin{proof}
First, notice that if an agent owns two elements in the winning
solution, \aiedit{\RP does not treat the agent differently from if she only
owned one element.}
Thus, the proof of Lemma \ref{lem:identifier-splitting} 
also shows that self-division can never be beneficial for an agent,
and we can assume from now on that no agent will self-divide or split
identifiers.
Thus, each agent $i$ submits bids $(b_e, i)$ for all
elements $e \in \Owned{i}$. 
If the set $S^*\in \Feasible$ wins under \RP, agent $i$'s utility is 
\begin{eqnarray*}
\Pay{i}-c(S^* \cap \Owned{i})
& = & B^{-i} - \big(b(S^* \setminus \Owned{i}) + \LPen{r}{\len{S^*}}
  + c(S^* \cap \Owned{i})\big).
\end{eqnarray*}
Since $B^{-i}$ is a constant independent of the bids \dkedit{$b_e$} by agent
$i$, agent $i$'s utility is maximized when
$(b(S^* \setminus \Owned{i}) + \LPen{r}{\len{S^*}}
+ c(S^* \cap \Owned{i}))$ is minimized. But this is exactly the
quantity that \RP will minimize when agent $i$ submits truthful bids
for all her elements; hence, truthfulness is a dominant strategy.
\end{proof}

The next theorem proves that an auctioneer with a reserve cost of $r$
faces no loss.

\begin{theorem} \label{thm:no-loss}
The sum of the payments made by \RP to agents never exceeds $r$.
\end{theorem}

\begin{proof}
Because we already proved that \RP is false-name-proof, we can without
loss of generality identify \dkedit{$c_e$ and $b_e$} for each element $e$.
When \dkedit{\LEN} agents are part of the winning set $S^*$, the payment to agent
$i$ is
\begin{eqnarray*}
\Pay{i}
& = & B^{-i} - \big(c(S^* \setminus \Owned{i}) + \LPen{r}{\LEN}\big)\\
& \leq & r - \big(c(S^* \setminus \Owned{i}) + r - \frac{r}{2^{\LEN-1}}\big)\\
& \dkedit{\leq} & \frac{r}{2^{\LEN-1}}.
\end{eqnarray*}
Thus, the sum of all payments to agents $i$ is at most
$w \cdot \frac{r}{2^{w-1}} \leq r$.
\end{proof}

Since the reserve cost mechanism does not always purchase a feasible set, 
we cannot analyze its frugality ratio in the sense of Definition \ref{def:frugality}. 
(The definition is based on the assumption that the mechanism always purchases a set.) 
Nevertheless, if the auctioneer already has a way of performing the task 
using a single agent of cost $r$, 
(such as a direct edge with higher cost in a network), 
we can derive bounds on the frugality ratio of \dkedit{the} \RP mechanism. 
These bounds cannot be taken as actual hard guarantees, 
since we need to assume that the auctioneer was ``lucky'' 
in choosing the right reserve cost.

Specifically, assume that the auctioneer chose a reserve cost 
$r \leq \dkedit{2^n} \cdot \max_{e\in{S}}c(T_{e})$, where $S$ is the cheapest solution, 
and the sets $T_e$ are defined by the third constraint of Definition~\ref{def:chnash}. 
Since the total payment of \RP does not exceed $r$ by Theorem~\ref{thm:no-loss}, 
and $\ChNash{\costs} \geq \max_{e \in S} c(T_e)$ by Inequality \ref{eqn:nash-bound}, 
we obtain an upper bound of \dkedit{$O(2^n)$} on the frugality ratio, 
matching that of \EP. 
More generally, if the auctioneer chooses an $r \leq f(n) \cdot \max_{e\in{S}}c(T_{e})$, 
then the frugality ratio of the mechanism is $O(f(n))$.

\subsection{Experiments} \label{sec:experiments}
We complement the analysis of the previous section 
with experiments for shortest $s$-$t$ path auctions on random graphs. 
Our simulation compares the payments of \RP with VCG, 
under the assumption that there is in fact no false-name manipulation 
and each agent owns one edge. 
Thus, we evaluate the overpayment caused by preventing false-name manipulation. 

Since some of our graphs have monopolies, 
we modify VCG by introducing a reserve cost $r$. 
Thus, if $S^*$ is the cheapest solution with respect to the cost, 
the reserve-cost VCG mechanism (\RVCG) only purchases a path when $c(S^*) \leq r$. 
In that case, the payment to each edge $e \in S^*$ is 
$\Pay{e} = \min (r, c(S^{-e})) - c(S^{*} \setminus \SET{e})$, 
where $S^{-e}$ is the cheapest solution not containing $e$.

Our generation process for random graphs is as follows: 
$40$ nodes are placed independently and uniformly 
at random in the unit square $[0,1]^2$. 
Then, 200 independent and uniformly random node pairs are connected
with edges.\footnote{%
We also ran simulations 
on random small-world networks~\cite{watts:nature:1998}.
\suggestedcut{
Many real-world networks in which routes must be purchased, 
such as truck routes, railroad tracks, natural gas pipeline networks,
or computer networks, are known to be small-world networks.
}
Our results for small-world networks are qualitatively similar, and we
therefore focus on the case of uniformly random networks here.
}
The cost of each edge $e$ is its Euclidean length. 
We evaluate $100$ random trials; in each, we seek to buy a path
between two randomly chosen nodes.  
While the number of nodes is rather small compared to the real-world
networks on which one would like to run auctions, it is dictated by
the computational complexity of the mechanisms we study. 
Larger-scale experiments are a fruitful direction for future work. 

Figure~\ref{fig:social surplus} shows the average social surplus 
(the difference between the reserve cost and the true cost 
incurred by edges on the chosen path, $r-\sum_{e \in S^*} c_{e}$) 
in \RP and \RVCG, as well as the ratio between the two, 
when varying the reserve cost $r \in [0,3.5]$. 
The social surplus for both increases roughly linearly 
under both mechanisms. 
While the plot shows some efficiency loss by using \RP, 
\dkedit{the efficiency} is always within a factor of about 60\% for
our instances, and on average around 80\%. 

Figure~\ref{fig:payments} illustrates the average payments of the auctioneer.
Clearly, small reserve costs lead to small payments, 
and when the reserve costs are less than $1.8$, 
the payment of \RP is in fact smaller than that of \RVCG.
As the reserve cost $r$ increases, \RVCG's payments converge, 
while those of \RP keep increasing almost linearly.
The reason is that the winning path in \RP tends to have fewer edges
than other competing paths, and is thus paid an increased bonus 
as $r$ increases. 
We would expect such behavior to subside 
as there are more competing paths with the same number of edges.

%
\begin{figure}[htb]
  \begin{minipage}[t]{0.5\linewidth}
    \centering
    \includegraphics[width=\linewidth]{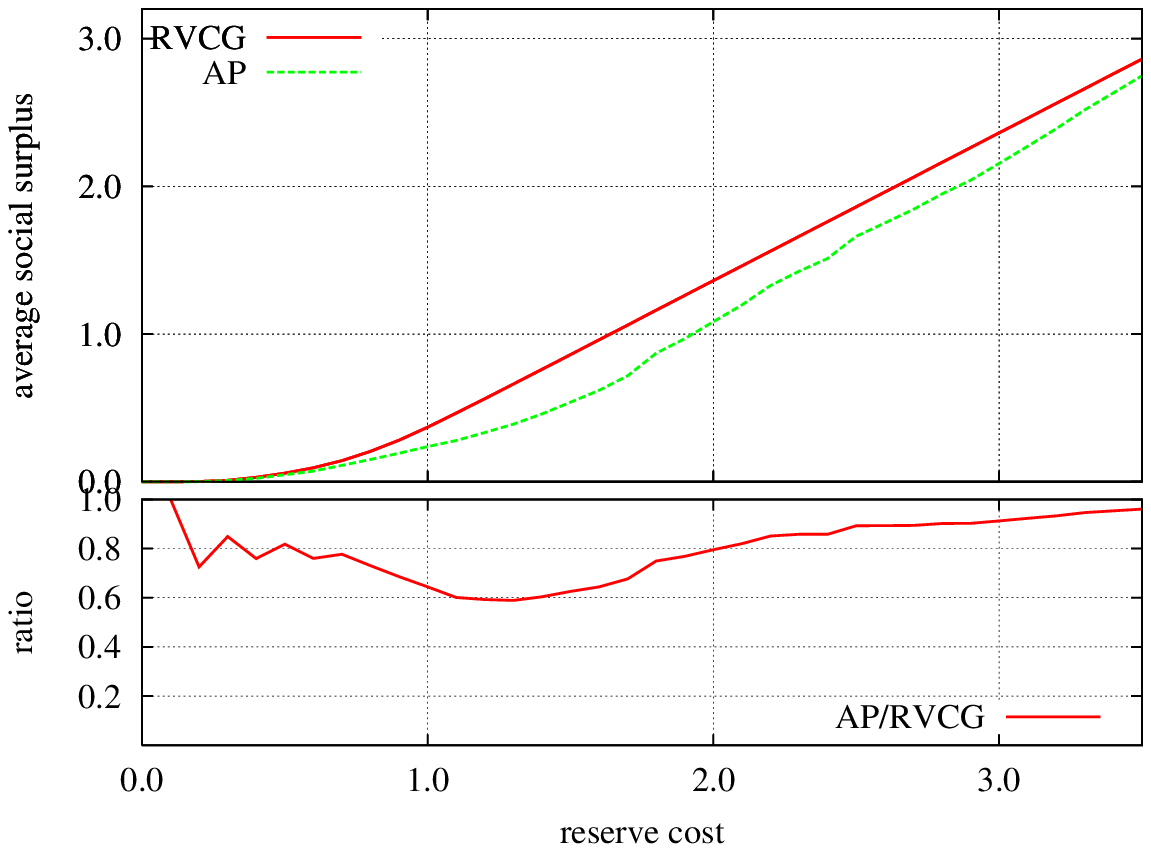}
    \caption{\dkedit{Social surplus}.}
    \label{fig:social surplus}
  \end{minipage}
  \begin{minipage}[t]{0.5\linewidth}
    \centering
    \includegraphics[width=\linewidth]{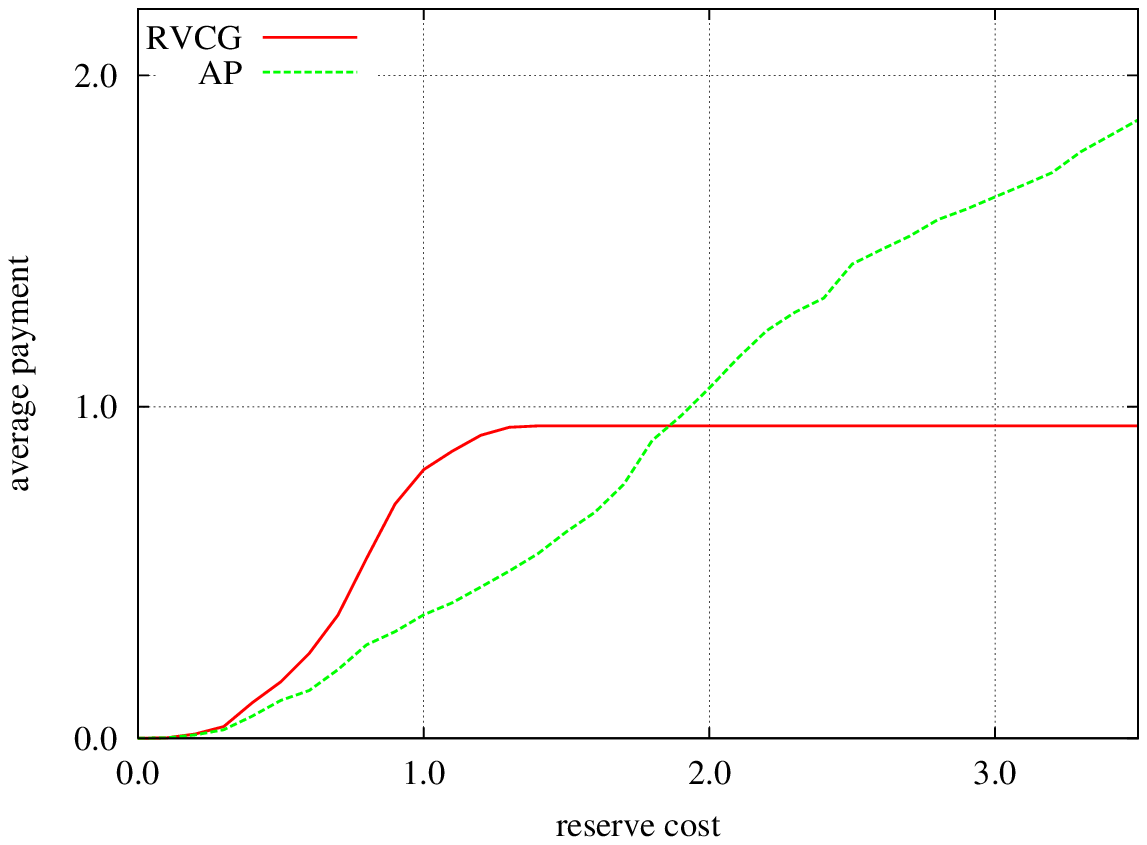}
    \caption{\dkedit{Payments}.}
    \label{fig:payments}
  \end{minipage}
\end{figure}
\section{Concluding remarks}
In this paper, 
we initiated the investigation of false-name-proof mechanisms for hiring a team of agents.
In this model, the structure of the set system may not be completely known to the auctioneer.
We first presented a mechanism \EP based on exponential multiplicative penalties, 
which always buys a solution, 
but is false-name-proof only when each agent has exactly one element. 
We proved that \EP has a frugality ratio of \dkedit{$2^n$}. 
This is within a \dkedit{constant} factor of optimal for all classes of set systems, 
as we also proved a lower bound of $\Omega(2^n)$ for all false-name-proof mechanisms. 

We also presented an \dkedit{alternative} mechanism \RP  
with exponential additive penalties and a reserve cost, 
which is false-name-proof 
even when each agent has multiple elements.
We evaluated \RP experimentally; while it has smaller social surplus
compared to VCG, the difference is bounded by
small multiplicative constants in all of our experiments.
The payments of \RP \dkedit{are} smaller than \dkedit{those of} the
VCG mechanism when the reserve cost is small. 
Although the payments increase linearly in the reserve cost, 
they never exceed the reserve cost. 

It remains open whether there is a mechanism which always purchases a solution, 
and is false-name-proof even when each agent has multiple elements. 
This holds even for such seemingly simple cases as $s$-$t$ path auctions. 
It may be possible that no such mechanism exists, 
which would be an interesting result in its own right. 
The difficulty of designing false-name-proof mechanisms for hiring a team 
is mainly due to a lack of useful characterization results for 
incentive-compatible mechanisms when agents have multiple parameters. 
While a characterization of truthful mechanisms has been given by 
Rochet~\cite{rochet:87}, this condition is difficult to apply in practice. 
\suggestedcut{
Only recently did Lavi and Swamy \cite{lavi:swamy}
present incentive compatible mechanisms for scheduling on machines
based on the condition of Rochet.
}

\dkedit{
\subsubsection*{Acknowledgments}
We would like to thank several anonymous reviewers for helpful
feedback on a previous version of this paper.
}




\end{document}